\documentclass[11pt]{article}
\usepackage{geometry}
\usepackage{amsmath}
\usepackage{amssymb}
\usepackage{indentfirst}
\usepackage{mathdots}
\usepackage{cite}
\usepackage[colorlinks, linkcolor=blue, citecolor=blue]{hyperref}
\usepackage{mathrsfs}
\usepackage{booktabs}
\usepackage{enumitem}
\usepackage{authblk}
\usepackage{amsfonts,ntheorem}

\newtheorem{definition}{Definition}
\newtheorem{theorem}{Theorem}
\newtheorem{lemma}{Lemma}
\newtheorem{corollary}{Corollary}
\newtheorem{proposition}{Proposition}
\newtheorem{remark}{Remark}
\newtheorem{example}{Example}
\newenvironment{proof}{\noindent{\textbf{\emph{Proof.}}}}

\begin{document}
\title{Some results about double cyclic codes over $\mathbb{F}_{q}+v\mathbb{F}_{q}+v^2\mathbb{F}_{q}$}
\author[a]{Tenghui Deng \thanks{Corresponding author.\\{E-mail addresses: dth16@mails.tsinghua.edu.cn (Tenghui Deng), y-j@tsinghua.edu.cn (Jing Yang).}}}
\author[a]{Jing Yang }
\affil[a]{Department of Mathematics, Tsinghua University, Beijing 100084, PR China }
\renewcommand*{\Affilfont}{\small\it}

\date{March 10, 2021}

\maketitle

\begin{abstract}
Let $\mathbb{F}_{q}$ be the finite field with $q$ elements. This paper mainly researches the polynomial representation of double cyclic codes over $\mathbb{F}_{q}+v\mathbb{F}_{q}+v^2\mathbb{F}_{q}$ with $v^3=v$. Firstly, we give the generating polynomials of these double cyclic codes. Secondly, we show the generating matrices of them. Meanwhile, we get quantitative information related to them by the matrix forms. Finally, we investigate the relationship between the generators of double cyclic codes and their duals.

\end{abstract}



\section{Introduction}
The theory of error-correcting code plays a crucial role in Internet data transmission, satellite positioning and communication. The significance has become increasingly prominent with the promotion and popularization of a series of high technologies such as artificial intelligence and 5G technology. For most coding researchers, Cyclic Code is the main research object because of its good structure and it is easy to discover, understand and decode in the process of application.

Based on the more flexible algebraic structure, the coding theory over finite rings is becoming more active in recent research. Recently, Borges et al. proposed a class of new codes called $\mathbb{Z}_2\mathbb{Z}_4$-additive codes in \cite{Borges1}. Since then, this family of codes were significant from theory to application, and some generalizations were also studied deeply. Abualrub in \cite{Abualrub1,Abualrub2}, Aydogdu in \cite{Aydogdu1,Aydogdu2}, T. Roy in \cite{Roy} and Borges in \cite{Borges2} were remarkable works. Around 2014, Borges et al. investigated double cyclic codes' algebraic structures $\mathbb{Z}_2$. The authors determined the generator polynomials of this family of codes and their duals. It was a vital work for the research to double cyclic codes. Since then, numerous articles about double cyclic codes have appeared. Such as: On double cyclic codes over $\mathbb{Z}_4$(\cite{Jian}), $\mathbb{Z}_4$-Double Cyclic Codes Are Asymptotically Good(\cite{Jian2}), Double $\lambda$-constacyclic codes over finite chain rings(\cite{Wang}) and Double cyclic codes over $\mathbb{F}_q+u\mathbb{F}_q+u^2\mathbb{F}_q$(\cite{Yao}). Some results on linear codes over $\mathbb{F}_p+u\mathbb{F}_p+u^2\mathbb{F}_p$(\cite{Jian3}) and Some results on $\mathbb {Z} _p\mathbb {Z} _p [v] $-additive cyclic codes(\cite{Lin2}) are also relevant works.

For $\mathbb{F}_{q}+v\mathbb{F}_{q}+v^2\mathbb{F}_{q}$ with $v^3=v$, there also were many results about this non-chain ring, A. Melakhessou et al.: On codes over $\mathbb{F}_{q}+v\mathbb{F}_{q}+v^2\mathbb{F}_{q}$ in \cite{Melakhessou}, Fanghui Ma et al.: Constacyclic codes over the ring and their applications of constructing new non-binary quantum codes in \cite{Fanghui}, Minjia Shi et al.: Skew cyclic codes over $\mathbb{F}_{q}+v\mathbb{F}_{q}+v^2\mathbb{F}_{q}$ in \cite{Shi}.

In this paper, applying the methods developed by Borges et al. and Gao et al., we get some algebraic structures of double cyclic codes over $\mathbb{F}_q+v\mathbb{F}_q+v^2\mathbb{F}_q$ and give some example for these double cyclic codes.

The rest of the article is organized as follows. In Section 2, we provide some necessary preliminaries about the polynomial theory over $\mathbb{F}_{q}+v\mathbb{F}_{q}+v^2\mathbb{F}_{q}$. In Section 3, we present the definition of double cyclic codes and set up some structural properties of double cyclic codes over $\mathbb{F}_{q}+v\mathbb{F}_{q}+v^2\mathbb{F}_{q}$. In Section 4, we give the generating matrix forms of these double cyclic codes. Simultaneously, utilizing these generating matrix forms, we get some quantitative information relating to double cyclic codes. In Section 5, we obtain generators' relationship between double cyclic codes and their duals.

\section{Preliminaries}
Throughout this paper, let $R$ denote $\mathbb{F}_{q}+v\mathbb{F}_{q}+v^2\mathbb{F}_{q}$, where $v^3=v$ and $\mathbb{F}_q$ be a finite field of a characteristic odd prime. From the knowledge of finite rings, we know that $R$ is equivalent to the quotient ring $\mathbb{F}_q[x]/\left\langle v^3-v\right\rangle$ and is also a finite commutative ring with identity. These indicated that $R$ is a principal ring that has only three non-trivial maximal ideals $\left\langle 1-v \right\rangle , \left\langle v \right\rangle , \left\langle 1+v \right\rangle $. Consequently, from the Chinese Remainder Theorem, we get $$R=R/\left\langle v\right\rangle  \oplus R/\left\langle 1-v\right\rangle  \oplus R/\left\langle 1+v\right\rangle .$$ Let $v_1=1-v^2, v_2=\frac{v+v^2}{2}, v_3=\frac{v^2-v}{2}$, we have $(v_1,v_2,v_3)=(1,v,v^2) \left(\begin{smallmatrix} 1&0&0 \\ 0&\frac{1}{2}&-\frac{1}{2} \\-1&\frac{1}{2}&\frac{1}{2} \end{smallmatrix}\right)$. Since $\det\left(  \left(\begin{smallmatrix} 1&0&0 \\ 0&\frac{1}{2}&-\frac{1}{2} \\-1&\frac{1}{2}&\frac{1}{2} \end{smallmatrix}\right) \right)  \ne 0$, this means that $\left\lbrace v_1,v_2,v_3 \right\rbrace $ is also a basis for $R$. Note that $v_1, v_2, v_3$ are all orthogonal idempotent elements in $R$. Then $R$ can be decomposed into $$R=Rv_1 \oplus Rv_2 \oplus Rv_3=\mathbb{F}_qv_1 \oplus \mathbb{F}_qv_2 \oplus \mathbb{F}_qv_3.$$ 

For $r \in R$, let $r=\sum_{i=1}^{3}r_{v_i}v_i$ with $r_{v_i} \in \mathbb{F}_q, i=1,2,3,$ define three projections as $P_{v_i}: r=\sum_{i=1}^{3}r_{v_i}v_i \mapsto r_{v_i}, i=1,2,3$. Then $P_{v_i}, i=1,2,3$ are $\mathbb{F}_q$-alegbra homomorphism. $\forall n \in \mathbb{N}$, we extend these maps from $R$ to $R^n$ naturally. Let $P_{v_i}: R^n \to \mathbb{F}^n_q, \text{ } (r_1,\dots r_n) \mapsto ((r_1)_{v_i},\dots (r_n)_{v_i})$, $i=1,2,3$. They still are $\mathbb{F}_q$-algebra homomorphism. In the same way, these maps can also expand to polynomial rings over $R$.

For each polynomial $r(x) \in R[x]$, considering the commutativity of $\left\lbrace  v_1,v_2,v_3 \right\rbrace $ and $x$, we can decompose the coefficients into standard bases and merge the homologous terms properly, then we can get the unique decomposition $r(x)=\sum_{i=1}^{3}r_{v_i}(x)v_i$ about these set of standard bases. Define three maps $P_{v_i}:R[x] \rightarrow \mathbb{F}_q[x] \text{ } r(x)=\sum_{i=1}^{3}r_{v_i}(x)v_i \mapsto r_{v_i}(x),i =1,2,3$. For $a(x),b(x) \in R[x]$, let $a(x)=\sum_{i=1}^{3}a_{v_i}(x)v_i, b(x)=\sum_{i=1}^{3}b_{v_i}(x)v_i$, we have that
$$\begin{cases}
a(x)+b(x) &=\sum_{i=1}^{3}(a_{v_i}(x)+b_{v_i}(x))v_i, \\
a(x)b(x)  &=\sum_{i=1}^{3}(a_{v_i}(x)b_{v_i}(x))v_i.  \\
\end{cases}$$ 
This mean that the projections $P_{v_i},i=1,2,3$ into $R[x]$ are also $\mathbb{F}_q[x]$-homomorphism. From the above station, we obtain that $R[x]=(\oplus_{i=1}^{3}\mathbb{F}_qv_i)[x]=\oplus_{i=1}^{3} \mathbb{F}_{q}[x]v_i$.

About the divisibility between any two elements in $R[x]$, we have 
\begin{lemma}
Let $a(x),b(x) \in R[x]$ with $a(x)=\sum_{i=1}^{3}a_{v_i}(x)v_i, b(x)=\sum_{i=1}^{3}b_{v_i}(x)v_i$. Then $a(x)|b(x)$ in $R[x]$ if and only if $a_{v_i}(x)|b_{v_i}(x)$, $i=1,2,3$ in $\mathbb{F}_{q}[x]$.
\end{lemma}
\begin{proof}
If $a(x)|b(x)$, we can set $b(x)=e(x)a(x)$ with $e(x) \in R[x]$. Hence $$\sum_{i=1}^{3}b_{v_i}(x)v_i=\left( \sum_{i=1}^{3}e_{v_i}(x)v_i\right) \left( \sum_{i=1}^{3}a_{v_i}(x)v_i\right) =\sum_{i=1}^{3}(e_{v_i}(x)a_{v_i}(x))v_i.$$
Then we have $b_{v_i}(x)=e_{v_i}(x)a_{v_i}(x), i=1,2,3$. This means that $a_{v_i}(x)|b_{v_i}(x), i=1,2,3$ in $\mathbb{F}_{q}[x]$.

On the contrary, due to $a_{v_i}(x)|b_{v_i}(x),i=1,2,3$, we can set $b_{v_i}(x)=e_{v_i}(x)a_{v_i}(x)$ with $e_{v_i}(x) \in \mathbb{F}_q[x]$ $i=1,2,3$. Then $$b(x)=\sum_{i=1}^{3}b_{v_i}(x)v_i=\sum_{i=1}^{3}e_{v_i}(x)a_{v_i}(x)v_i=\left( \sum_{i=1}^{3}e_{v_i}(x)v_i\right) \left( \sum_{i=1}^{3}a_{v_i}(x)v_i\right) =\left( \sum_{i=1}^{3}e_{v_i}(x)v_i\right) a(x).$$
Let $e(x)=\sum_{i=1}^{3}e_{v_i}(x)v_i \in R[x]$. It shows that $a(x)|b(x)$ in $R[x]$.
\end{proof}
\begin{remark}
From Lemma 1, it also follows that $P_{v_i}(\frac{b(x)}{a(x)})=e_{v_i}(x)=\frac{P_{v_i}(b(x))}{P_{v_i}(a(x))}$, $i=1,2,3,$ when $a(x),b(x) \in R[x]$, $a(x)|b(x)$.
\end{remark}
For the sake of the divisible form, we also have the following result about the greatest common divisor for any two elements in $R[x]$. 
\begin{lemma}
Let $a(x),b(x) \in R[x]$ with $a(x)=\sum_{i=1}^{3}a_{v_i}(x)v_i, b(x)=\sum_{i=1}^{3}b_{v_i}(x)v_i$. Then, in $R[x]$, we have $$\gcd(a(x),b(x))=\gcd\left( \sum_{i=1}^{3}a_{v_i}(x)v_i,\sum_{i=1}^{3}b_{v_i}(x)v_i\right) =\sum_{i=1}^{3}\gcd(a_{v_i}(x),b_{v_i}(x))v_i,$$ 
where the $\gcd(-,-)$ on the right side of the formula represents the greatest common factor of polynomials in $\mathbb{F}_q[x]$. 
\end{lemma}
\begin{proof}
Denote $\sum_{i=1}^{3}\gcd(a_{v_i}(x),b_{v_i}(x))v_i$ by $m(x)$, then $m(x) \in R[x]$. From the conventional polynomial theory, we know that $\gcd(a_{v_i}(x),b_{v_i}(x))|a_{v_i}(x)$, $\gcd(a_{v_i}(x),b_{v_i}(x))|b_{v_i}(x)$, $i=1,2,3$. For Lemma 1, we have $m(x)|a(x)$, $m(x)|b(x)$ . This follows that $m(x)$ is a common divisor of $a(x)$ and $b(x)$.

Let $r(x) \in R[x]$ with $r(x)|a(x)$, $r(x)|b(x)$. We can set $r(x)=\sum_{i=1}^{3}r_{v_i}(x)v_i$. For the same reason of Lemma 1, we have $r_{v_i}(x)|a_{v_i}(x)$, $r_{v_i}(x)|b_{v_i}(x)$, $i=1,2,3$. Thus $r_{v_i}(x)|\gcd(a_{v_i}(x),b_{v_i}(x))$ in $\mathbb{F}_q[x]$, $i=1,2,3$. Consequently, $r(x)|m(x)$. These reveal that every common divisor of $a(x)$ and $b(x)$ is a divisor of $m(x)$. 

In conclusion, we have $\sum_{i=1}^{3}\gcd(a_{v_i}(x),b_{v_i}(x))v_i=m(x)=\gcd(a(x),b(x))$.     	
\end{proof}
\begin{remark}
From the process of above proof, we also have $$P_{v_i}(\gcd(a(x),b(x)))=m_{v_i}(x)=\gcd(P_{v_i}(a(x)),P_{v_i}(b(x))), i=1,2,3.$$
\end{remark}

Summarize the above statement, we have 
$$R[x]/\left\langle r[x] \right\rangle =(\oplus_{i=1}^{3} \mathbb{F}_{q}[x]v_i) \bigg/ \left\langle  \sum_{i=1}^{3} r_{v_i}(x)v_i \right\rangle  = \oplus_{i=1}^{3} (\mathbb{F}_{q}[x]/ \left\langle r_{v_i}(x) \right\rangle )v_i,$$
where $r(x) \in R[x]$ with $r(x)=\sum_{i=1}^{3}r_{v_i}(x)v_i$.
 
About single cyclic codes over $R$, we list some conclusions, see \cite{Melakhessou}, which will be used to obtain our results.
 
For a linear code $C$ of length $n$ over $R$, define $$\begin{cases}
C_1=\left\lbrace a \in \mathbb{F}_q|\exists b,c \in \mathbb{F}_q \text{ } \text{such that} \text{ } av_1+bv_2+cv_3 \in C \right\rbrace, \\
C_2=\left\lbrace b \in \mathbb{F}_q|\exists a,c \in \mathbb{F}_q \text{ } \text{such that} \text{ } av_1+bv_2+cv_3 \in C \right\rbrace, \\
C_3=\left\lbrace c \in \mathbb{F}_q|\exists a,b \in \mathbb{F}_q \text{ } \text{such that} \text{ } av_1+bv_2+cv_3 \in C \right\rbrace. \\
\end{cases}$$
\noindent
\textbf{Facts 1}(\cite{Melakhessou}) Let $C=v_1C_1 \oplus v_2C_2 \oplus v_3C_3$ be a linear code of length $n$ over $R$. Then $C$ is a cyclic code of length $n$ over $R$ if and only if $C_i,i=1,2,3$ are cyclic codes of length $n$ over $\mathbb{F}_q$.

\noindent
\textbf{Facts 2}(\cite{Melakhessou}) If $C=v_1C_1 \oplus v_2C_2 \oplus v_3C_3$ is a cyclic code of length $n$ over $R$, then $|C|=\sum_{i=1}^{3}|C_i|$

\noindent
\textbf{Facts 3}(\cite{Melakhessou}) If $C=v_1C_1 \oplus v_2C_2 \oplus v_3C_3$ is a cyclic code of length $n$ over $R$, then $$C^{\bot}=v_1C_1^{\bot} \oplus v_2C_2^{\bot} \oplus v_3C_3^{\bot} \text{ } \text{and} \text{ } |C^{\bot}|=\sum_{i=1}^{3}|C_i^{\bot}|.$$ 

\section{Double cyclic codes}
\begin{definition}
Let $C$ be an $R$-linear code of length $m+n$. Code $C$ called a double cyclic code of length $(m,n)$ over $R$ if for any $$c=(c^1_0,\cdots,c^1_{m-2},c^1_{m-1}|c^2_0,\cdots,c^2_{n-2},c^2_{n-1}) \in C,$$ implies the cyclic shift $$\mathsf{T}(c)=(c^1_{m-1},c^1_0,\cdots,c^1_{m-2}|c^2_{n-1},c^2_0,\cdots,c^2_{n-2}) \in C.$$  
\end{definition}
From the definition of double cyclic codes over $R$, we obtain that the double cyclic code $C$ can be viewed as an $R$-submodule of $R^m \times R^n$.

Let $c, d \in R^m \times R^n$ with $c=(c^1_0,\cdots,c^1_{m-1}|c^2_0,\cdots,c^2_{n-1}), d=(d^1_0,\cdots,d^1_{m-1}|d^2_0,\cdots,d^2_{n-1})$, we define the inner product of these two elements for $\left\langle c,d \right\rangle \triangleq \sum_{i=0}^{m-1} c^1_id^1_i+\sum_{j=0}^{n-1} c^2_jd^2_j$.

For the double cyclic code $C$, the dual code is defined as $$C^{\bot}=\left\lbrace d \in R^m \times R^n|\left\langle d,c \right\rangle =0, \text{ } \forall c \in C\right\rbrace.$$

For $R^m \times R^n$, define two coordinate projections as  $$\begin{cases}
P_m:R^m \times R^n \to R^m, &(r^1_1, \cdots, r^1_m|r^2_1, \cdots, r^2_n) \mapsto (r^1_1, \cdots, r^1_m), \\
P_n:R^m \times R^n \to R^n, &(r^1_1, \cdots, r^1_m|r^2_1, \cdots, r^2_n) \mapsto (r^2_1, \cdots, r^2_n). \\
\end{cases}$$
For $f(x) \in \mathbb{F}_q[x]$, owing to $\sum_{i=1}^{3}v_i=1$, we can view it as $f(x)=\sum_{i=1}^{3}f(x)v_i$. This guarantees that for every $f(x) \in \mathbb{F}_q[x]$, the quotient ring $R[x]/\left\langle f(x) \right\rangle$ is well defined. Then for $x^m-1, x^n-1 \in \mathbb{F}_q[x]$, let
$$R_m[x] \triangleq R[x]/\left\langle x^m-1 \right\rangle; R_{m,n}[x] \triangleq \left( R[x]/\left\langle x^m-1\right\rangle \right) \times \left( R[x]/\left\langle x^n-1\right\rangle \right); R_n[x] \triangleq R[x]/\left\langle x^n-1\right\rangle.$$
The multiplication of $R[x]$ can induce the action  of $R[x]$ on  $R_m[x]$, $R_{m,n}[x]$, $R_n[x]$ naturally. The rings $R_m[x]$, $R_{m,n}[x]$, $R_n[x]$ become an $R[x]$-module toward with this action. There are also two coordinate projectios for the bijection between  $R^m \times R^n$ to $R_{m,n}[x]$ given by $$(c^1_0,c^1_1, \cdots ,c^1_{m-1}|c^2_0,c^2_1, \cdots ,c^2_{n-1}) \mapsto (c^1_0+c^1_1x+ \cdots +c^1_{m-1}x^{m-1}|c^2_0+c^2_1x+ \cdots +c^2_{n-1}x^{n-1}).$$
They are
$$\begin{cases}
P_m: R_{m,n}[x] \rightarrow R_m[x]&(a(x)|b(x)) \mapsto a(x), \\
P_n: R_{m,n}[x] \rightarrow R_n[x]&(a(x)|q(x)) \mapsto b(x). \\
\end{cases}$$
Then $P_m, P_n$ are still $R[x]$-module homomorphism. Similarly, basing on the one-to-one correspondence between $R^m \times R^n$ and $R_{m,n}[x]$, it reveals the fact that $C$ is a double cyclic code of length $(m,n)$ over $R$ if and only if the corresponding polynomial sets is an $R[x]$-submodule of $R_{m,n}[x]$. Then we can be concern about $R[x]$-submodule to $R_{m,n}[x]$ and regard $R$-double cyclic codes as the $R[x]$-submodule of $R_{m,n}[x]$. At the same time, we get two types of projections: the canonical projections $P_{v_i}$, $i=1,2,3$ and the coordinate projections $P_m, P_n$. In the future, unless otherwise specified, someone mathematical object appear the subscripts of $m,n,v_1,v_2,v_3$, which means that we use their corresponding projections by default.

Let $l=[m,n]$(i.e. $l$ denote the least common multiple of $m$ and $n$).
\begin{proposition}
If $C$ is a double cyclic code of length $(m,n)$ over $R$, then the dual code $C^{\bot}$ is also an $R$-double cyclic code of the same length.
\end{proposition}
\begin{proof}
Let $C$ be an $R$-double cyclic code of length $(m,n)$ and $d=(d^1_0, \cdots ,d^1_{m-1}|d^2_0, \cdots ,d^2_{n-1})$ be a codeword in $C^{\bot}$. By the definition of $R$-double cyclic codes, we need to explain that the cyclic shift codeword $\mathsf{T}(d) \in C^{\bot}$. This means that we need to prove $\left\langle \mathsf{T}(d),c\right\rangle=0$ for all codeword $c \in C$. 

Let $c$ be any codeword of $C$. By mathematical induction, we can easily get $\mathsf{T}^{l-1}(c)=\mathsf{T}^{l-2}(\mathsf{T}(c)) \in C$. Note that $\mathsf{T}^l(c)=c$. Fix $c=(c^1_0,\cdots ,c^1_{m-1}|c^2_0,\cdots ,c^2_{n-1})$ and write the specific form of $T^{l-1}(c)$, we get $\mathsf{T}^{l-1}(c)=(c^1_1, \cdots ,c^1_{m-1},c^1_0|c^2_1,\cdots ,c^2_{n-1},c^2_0)$. Since $d \in C^{\bot}$ and $c \in C$, we have 
$$0=\left\langle d, \mathsf{T}^{l-1}(c) \right\rangle = d^1_0c^1_1+ \cdots +d^1_{m-2}c^1_0+d^1_{m-1}c^1_0 + d^2_0c^2_1+ \cdots +d^2_{m-2}c^2_0+d^2_{m-1}c^2_0 = \left\langle \mathsf{T}(d),c\right\rangle,$$
which implies that $\mathsf{T}(d) \in C^{\bot}$. Therefore, $C^{\bot}$ is also an $R$-double cyclic code of length $(m,n)$.
\end{proof}
\begin{proposition}
Let $C$ be an $R$-double cyclic code of length $(m,n)$. Then there exist polynomials $\iota(x),\ell(x) \in R_m[x], \theta(x) \in R_n[x]$ with
$\iota(x)=\sum_{i=1}^{3}\iota_{v_i}(x)v_i, \ell(x)=\sum_{i=1}^{3}\ell_{v_i}(x)v_i, \theta(x)=\sum_{i=1}^{3}\theta_{v_i}(x)v_i$ such that $C$ has the forms of 
$$C=\left\langle (\iota(x)|0),(\ell(x)|\theta(x)) \right\rangle=\left\langle \left( \sum_{i=1}^3\iota_{v_i}(x)v_i|0\right) ,\left( \sum_{i=1}^3\ell_{v_i}(x)v_i| \sum_{i=1}^3 \theta_{v_i}(x)v_i\right) \right\rangle,$$
where $\iota_{v_i}(x)|x^m-1, i=1,2,3$ and $\theta_{v_i}(x)|x^n-1, i=1,2,3$.
\end{proposition}
\begin{proof}
Considering the coordinate projection $P_n : R_{m,n}[x] \to R_n[x] \text{ } (a(x)|b(x)) \mapsto b(x)$. Clearly, $P_n$ is an $R[x]$-module homomorphism. While the double cyclic code $C$ is an $R[x]$-submodule of $R_{m,n}[x]$. Then $P_n(C)$ is an $R[x]$-submodule of $R_n[x]$, which means that $P_n(C)$ is an ideal of $R_n[x]$. From the structure of $R_n[x]$, we can set $P_n(C)=\left\langle \theta(x) \right\rangle $ with $\theta(x)=\sum_{i=1}^{3}\theta_{v_i}(x)v_i$, where $\theta_{v_i}(x)|x^n-1$, $i=1,2,3$. Note that $\mathrm{Ker}(P_n|_{C})=\left\lbrace (c^1(x)|c^2(x)) \in C| c^2(x)=0 \right\rbrace$.  Define the set $I=\left\lbrace c^1(x) \in R_m[x]| (c^1(x)|0) \in \mathrm{Ker}(P_n|_{C}) \right\rbrace$. Obviously, $I$ is an ideal of $R_m[x]$. The same reasoning applies to the case of $R_m[x]$, we can set $I=\left\langle \iota(x) \right\rangle $ with $\iota(x)=\sum_{i=1}^{3}\iota_{v_i}(x)v_i$, where $\iota_{v_i}(x)|x^m-1$, $i=1,2,3$. For any element $(c^1(x)|0) \in \mathrm{Ker}(P_n|_{C})$, we have $c^1(x) \in I$ and there exists a polynomial $m(x) \in R[x]$ such that $c^1(x)=m(x)\iota(x)$. Thus $(c^1(x)|0)=m(x)(\iota(x)|0)$, which implies that $\mathrm{Ker}(P_n|_{C})$ is an $R[x]$-submodule of $C$ generated by $(\iota(x)|0)$. Therefore, by the First Isomorphism Theorem of Module Homomorphism, we have $C/\mathrm{Ker}(P_n|_{C}) \cong \mathrm{Img}(P_n|_{C})=P_n(C)=\left\langle \theta(x)\right\rangle $. For $\theta(x)$, let $(\ell(x)|\theta(x)) \in C$ with $P_n(\ell(x)|\theta(x))=\theta(x)$, where $\ell(x)=\sum_{i=1}^{3}\ell_{v_i}(x)v_i$, $\ell_{v_i}(x) \in \mathbb{F}_q[x]$, $i=1,2,3$. We show that $C=\left\langle (\iota(x)|0),(\ell(x)|\theta(x)) \right\rangle$ in the following proof.

Let $c(x)=(c^1(x)|c^2(x)) \in C$, we have $c^2(x)=P_n((c^1(x)|c^2(x))) \in P_n(C)$. Consequently, $\exists s(x) \in R_n[x]$ such that $c^2(x)=s(x)\theta(x)$. While $$c(x)-s(x)(\ell(x)|\theta(x))=(c^1(x)|c^2(x))-s(x)(\ell(x)|\theta(x))=(c^1(x)-s(x)\ell(x)|0) \in Ker(P_n|_{C}).$$
Then $\exists r(x) \in R_m[x]$ such that $c^1(x)-s(x)\ell(x)=r(x)\iota(x)$ and $(c^1(x)-s(x)\ell(x)|0)=r(x)(\iota(x)|0)$. Hence $$c(x)=(c^1(x)|c^2(x))=r(x)(\iota(x)|0)+s(x)(\ell(x)|\theta(x)).$$ These indicate that $C$ is finitely generated by $\left\lbrace(\iota(x)|0), (\ell(x)|\theta(x)) \right\rbrace $.     
\end{proof}
\begin{lemma}
Let $C \in \mathfrak{C}_{m,n}(R)$ with 
$$C=\left\langle (\iota(x)|0),(\ell(x)|\theta(x)) \right\rangle=\left\langle \left( \sum_{i=1}^3\iota_{v_i}(x)v_i|0\right) ,\left( \sum_{i=1}^3\ell_{v_i}(x)v_i| \sum_{i=1}^3 \theta_{v_i}(x)v_i\right) \right\rangle.$$ 
As the simplest forms of generator polynomials, we have $\deg(\ell_{v_i}(x)) < \deg(\iota_{v_i}(x))$, $i=1,2,3.$
\end{lemma}
\begin{proof}
If not, $\exists i \in\left\lbrace 1,2,3 \right\rbrace $ such that $\deg(\ell_{v_i}(x)) \geq \deg(\iota_{v_i}(x))$. There is no loss of generating in assuming $\deg(\ell_{v_3}(x)) \geq \deg(\iota_{v_3}(x))$. Fix $k=\deg(\ell_{v_3}(x))-\deg(\iota_{v_3}(x))$, then $k \ge 0$. Let $D=\left\langle \left( \sum_{i=1}^{3}\iota_{v_i}(x)v_i|0\right) ,\left( \sum_{i=1}^{2}\ell_{v_i}(x)v_i+(\ell_{v_3}(x)-x^k \iota_{v_3}(x))v_3|\sum_{i=1}^{3}\theta_{v_i}(x)v_i\right)  \right\rangle $. It is obvious that $\deg(\ell_{v_3}(x)-x^k \iota_{v_3}(x)) < \deg(\ell_{v_3}(x))$. Since the generators of $D$ belong to $C$, we have $D \subset C$. On the other hand $$( \sum_{i=1}^{3}\ell_{v_i}(x)v_i|\sum_{i=1}^{3}\theta_{v_i}(x)v_i) =( \sum_{i=1}^{2}\ell_{v_i}(x)v_i+(\ell_{v_3}(x)-x^k \iota_{v_3}(x))v_3|\sum_{i=1}^{3}\theta_{v_i}(x)v_i) +v_3x^k( \sum_{i=1}^{3}\iota_{v_i}(x)v_i|0) .$$
Then $\left( \sum_{i=1}^{3}\ell_{v_i}(x)v_i|\sum_{i=1}^{3}\theta_v(x)v_i\right)  \subset D$. This shows that $C \subset D$. Consequently, $D=C$. By decreasing $\deg(\ell_{v_3}(x))$, we end up with $\deg(\ell_{v_3}(x)) < \deg(\iota_{v_3}(x))$.
\end{proof}
\begin{lemma}
Let $C$ be a double cyclic code of length $(m,n)$ over $R$, and set 
$$C=\left\langle (\iota(x)|0),(\ell(x)|\theta(x)) \right\rangle=\left\langle \left( \sum_{i=1}^3\iota_{v_i}(x)v_i|0\right) ,\left( \sum_{i=1}^3\ell_{v_i}(x)v_i| \sum_{i=1}^3 \theta_{v_i}(x)v_i\right) \right\rangle. $$
Then $\iota_{v_i}(x)|\frac{x^n-1}{\theta_{v_i}(x)}\ell_{v_i}(x)$, $i=1,2,3$.
\end{lemma}
\begin{proof}
By the proof of Proposition 2, we know that $\mathrm{Ker}(P_n|_{C})=\left\langle (\iota(x)|0) \right\rangle$, where $P_n|_C$ is the second coordinate projection restricted to $C$. We concern about the codeword to $\frac{x^n-1}{\theta(x)} (\ell(x)|\theta(x))$. 

Since $\frac{x^n-1}{\theta(x)} (\ell(x)|\theta(x))=(\frac{x^n-1}{\theta(x)}\ell(x)|0) \in \mathrm{Ker}(P_n|_{C})$, we have $\iota(x)|\frac{x^n-1}{\theta(x)}\ell(x)$. From Lemma 1, we get $\iota_{v_i}(x)|\frac{x^n-1}{\theta_{v_i}(x)}\ell_v(x)$, $i=1,2,3$. 
\end{proof}
\begin{lemma}
Let $C$ be a double cyclic code of length $(m,n)$ over $R$, and set  
$$C=\left\langle (\iota(x)|0),(\ell(x)|\theta(x)) \right\rangle=\left\langle \left( \sum_{i=1}^3\iota_{v_i}(x)v_i|0\right) ,\left( \sum_{i=1}^3\ell_{v_i}(x)v_i| \sum_{i=1}^3 \theta_{v_i}(x)v_i\right) \right\rangle.$$
Then $\iota_{v_i}(x)|\frac{x^n-1}{\theta_{v_i}(x)}\gcd(\iota_{v_i}(x),\ell_{v_i}(x)), i=1,2,3$.
\end{lemma}
\begin{proof}
From Proposition 2, we have $\theta_{v_i}(x)|x^n-1$, $i=1,2,3$. These follow that $\iota_{v_i}(x)|\frac{x^n-1}{\theta_{v_i}(x)} \iota_{v_i}(x)$, $i=1,2,3$. By Lemma 4, we have $\iota_{v_i}(x)|\frac{x^n-1}{\theta_{v_i}(x)}\ell_{v_i}(x)$, $i=1,2,3$. Therefore $$\iota_{v_i}(x)|\gcd\left( \frac{x^n-1}{\theta_{v_i}(x)} \iota_{v_i}(x), \frac{x^n-1}{\theta_{v_i}(x)}\ell_{v_i}(x)\right)  |\frac{x^n-1}{\theta_{v_i}(x)}\gcd(\iota_{v_i}(x),\ell_{v_i}(x)), i=1,2,3.$$ Consequently, $\iota_{v_i}(x)|\frac{x^n-1}{\theta_{v_i}(x)}\gcd(\iota_{v_i}(x),\ell_{v_i}(x)), i=1,2,3$. 
\end{proof}
\begin{definition}
Let $C$ be a double cyclic code of length $(m,n)$ over $R$. If $C$ is the direct product of $C_m$ and $C_n$, then $C$ is called a separable double cyclic code. 
\end{definition}
\begin{lemma}
If $C=\left\langle \left( \sum_{i=1}^3\iota_{v_i}(x)v_i|0\right) ,\left( \sum_{i=1}^3\ell_{v_i}(x)v_i| \sum_{i=1}^3 \theta_{v_i}(x)v_i\right) \right\rangle$ is a separable $R$-double cyclic code, then $\ell_{v_i}(x)=0, i=1,2,3$.
\end{lemma}
\begin{proof}
By the definition of separable $R$-double cyclic codes, the proof is straightforward.
\end{proof}

Combining with some propositions and lemmas in this section, we get the first significant theorem of this paper.
\begin{theorem}
Let $C$ be an $R$-double cyclic code of length $(m,n)$, and set $C,C^{\bot}$ have the forms of $$\begin{cases}
C&=\left\langle (\iota(x)|0),(\ell(x)|\theta(x)) \right\rangle =\left\langle \left( \sum_{i=1}^3\iota_{v_i}(x)v_i|0\right) ,\left( \sum_{i=1}^3\ell_{v_i}(x)v_i| \sum_{i=1}^3 \theta_{v_i}(x)v_i\right) \right\rangle  \\
C^{\bot}&=\left\langle (\overline{\iota}(x)|0),(\overline{\ell}(x)|\overline{\theta}(x)) \right\rangle =\left\langle \left( \sum_{i=1}^3\overline{\iota}_{v_i}(x)v_i|0\right),\left( \sum_{i=1}^3\overline{\ell}_{v_i}(x)v_i| \sum_{i=1}^3 \overline{\theta}_{v_i}(x)v_i\right) \right\rangle.  \\
\end{cases}$$
Then $\iota_{v_i}(x)|x^m-1$, $\theta_{v_i}(x)|x^n-1$, $i=1,2,3$ and
 
If $C$ is a separable $R$-double cyclic code, we have $\ell_{v_i}(x)=0, i=1,2,3$.

If $C$ is a free $R$-double cyclic code, we have $\begin{cases}
(1)\deg(\ell_{v_i}(x)) < \deg(\iota_{v_i}(x)), \\ (2) \iota_{v_i}(x)|\frac{x^n-1}{\theta_{v_i}(x)}\ell_{v_i}(x), \\ (3)\iota_{v_i}(x)|\frac{x^n-1}{\theta_{v_i}(x)}\gcd(\iota_{v_i}(x),\ell_{v_i}(x)), \\ \end{cases}$ 
$i=1,2,3.$
\end{theorem}

\section{Generating matrices}
\begin{proposition}
Let $\thickmuskip=0mu \medmuskip=0mu C=\left\langle (\iota(x)|0),(\ell(x)|\theta(x)) \right\rangle=\left\langle ( \sum_{i=1}^3\iota_{v_i}(x)v_i|0),( \sum_{i=1}^3\ell_{v_i}(x)v_i| \sum_{i=1}^3 \theta_{v_i}(x)v_i) \right\rangle$ be a double cyclic code of length $(m,n)$ over $R$. Then $C$ is permutation equivalent to a $\mathbb{F}_{q}$-linear code with generator matrix of the form $G=\begin{pmatrix}
G_1v_1\\G_2v_2\\G_3v_3\\
\end{pmatrix}$, where 
$$
G_i=\left( {\begin{array}{c|c}
\begin{matrix}
I_{m-\deg(\iota_{v_i}(x))} & A^1_i & A^2_i \\
0                       & B^1_i & B^2_i \\
0                       & 0          & 0           \\
		
\end{matrix} &
\begin{matrix}
0               & 0        & 0                            \\
B^3_i       & I_{k_i}  & 0                            \\
M^1_i       & M^2_i &I_{n-\deg(\iota_{v_i}(x))-k_i} \\
\end{matrix}
\end{array}} \right) , i=1,2,3$$
and $k_i=\deg(\iota_{v_i}(x))-\deg(\gcd(\iota_{v_i}(x),\ell_{v_i}(x))) \in \mathbb{N}$, $i=1,2,3$.
\end{proposition}
\begin{proof}
It follows that all of $C_{v_i}, i=1,2,3$ are double cyclic codes over $\mathbb{F}_q$ from the canonical projections. By the Proposition 8 of paper \cite{Diao}, we know that each of $C_{v_i}$ is permutation equivalent to a linear code over $\mathbb{F}_q$, and their generating Matrices have the following forms:
$$
G_i=\left( {\begin{array}{c|c}
\begin{matrix}
I_{m-\deg(\iota_{v_i}(x))} & A^1_i & A^2_i \\
0                       & B^1_i & B^2_i \\
0                       & 0          & 0           \\
	
\end{matrix} &
\begin{matrix}
0               & 0        & 0                            \\
B^3_i       & I_{k_i}  & 0                            \\
M^1_i       & M^2_i &I_{n-\deg(\iota_{v_i}(x))-k_i} \\
\end{matrix}
\end{array}} \right) , i=1,2,3,$$
where $B^1_i$ are full rank square matrices of size $k_i \times k_i$, $i=1,2,3$. By reducing the canonical projections to the double cyclic codes, we obtain that the matrix form of $C$ is $G$.
\end{proof}

From the generator matrix of the $R$-double cyclic codes, we have an easy computation as follows
\begin{corollary}
Let $C=\left\langle (\iota(x)|0),(\ell(x)|\theta(x)) \right\rangle=\left\langle (\sum_{i=1}^3\iota_{v_i}(x)v_i|0),(\sum_{i=1}^3\ell_{v_i}(x)v_i| \sum_{i=1}^3 \theta_{v_i}(x)v_i)\right\rangle$ be a double cyclic code of length $(m,n)$ over $R$. Then $C$ is an $\mathbb{F}_{q}$-linear code of dimension $3m+3n-\sum_{i=1}^3(\deg(\iota_{v_i}(x))+\deg(\theta_{v_i}(x)))$. 
\end{corollary}

\begin{proposition}
Let $\thickmuskip=0mu \medmuskip=0mu C=\left\langle (\iota(x)|0),(\ell(x)|\theta(x)) \right\rangle =\left\langle (\sum_{i=1}^3\iota_{v_i}(x)v_i|0),(\sum_{i=1}^3\ell_{v_i}(x)v_i| \sum_{i=1}^3 \theta_{v_i}(x)v_i)\right\rangle$ be a double cyclic code of length $(m,n)$ over $R$ and $$C^{\bot}=\left\langle (\overline{\iota}(x)|0),(\overline{\ell}(x)|\overline{\theta}(x)) \right\rangle =\left\langle \left( \sum_{i=1}^3\overline{\iota}_{v_i}(x)v_i|0\right),\left( \sum_{i=1}^3\overline{\ell}_{v_i}(x)v_i| \sum_{i=1}^3 \overline{\theta}_{v_i}(x)v_i\right) \right\rangle.$$
Then 
$$\begin{cases}
|C_m|=q^{3m+\sum_{i=1}^{3}k_i-\sum_{i=1}^{3}\deg(\iota_{v_i}(x))},&|C_n|=q^{3n-\sum_{i=1}^{3}\deg(\theta_{v_i}(x))} ,\\
|(C^{\bot})_m|=q^{\sum_{i=1}^{3}\deg(\theta_{v_i}(x))},           &|(C^{\bot})_n|=q^{\sum_{i=1}^{3}\deg(\theta_{v_i}(x))+\sum_{i=1}^{3}k_i},\\
\end{cases}$$
where $k_i=\deg(\iota_{v_i}(x))-\deg(\gcd(\iota_{v_i}(x),\ell_{v_i}(x))) \in \mathbb{N}$, $i=1,2,3$.
\end{proposition}
\begin{proof}
According to Proposition 2, we know that $C_m$ is generated by the polynomial $\gcd(\iota(x),\ell(x))$ and $C_n$ is generated by $\theta(x)$.
By the numbers of codeword to single cyclic codes over $R$, we obtain that   
$$\begin{cases}
|C_m|=\sum_{i=1}^{3}|(C_m)_{v_i}|=\sum_{i=1}^{3}q^{m-\deg(\iota_{v_i}(x),\ell_{v_i}(x))}=q^{3m+\sum_{i=1}^{3}k_i-\sum_{i=1}^{3}\deg(\iota_{v_i}(x))}; \\
|C_n|=\sum_{i=1}^{3}|(C_n)_{v_i}|=\sum_{i=1}^{3}q^{n-\deg(\theta_{v_i}(x))}=q^{3n-\sum_{i=1}^{3}\deg(\theta_{v_i}(x))}.
\end{cases}$$
	
From the generating matrix forms of $C$, we can efficiently work out the parity check matrix of $C$ is $H=\begin{pmatrix} H_1v_1\\H_2v_2\\H_3v_3\\ \end{pmatrix}$, where
$$H_i=\left( {\begin{array}{c|c}
\begin{matrix}
(A^1_i)^t  & I_{k_i} & 0                      \\
(A^2_i)^t & 0       & I_{\deg(\iota_{v_i}(x))-k_i}  \\                      
0             & 0       & 0                       \\
\end{matrix} &
\begin{matrix}
0                         & (B^1_i)^t  & (B^1_i)^t (M^2_i)^t           \\
0                         & (B^2_i)^t  & (B^2_i)^t (M^2_i)^t           \\
I_{\deg(\theta_{v_i}(x))} & (B^3_i)^t  & (M^1_i)^t+(B^3_i)^t (M^2_i)^t \\
\end{matrix}
\end{array}} \right), i=1,2,3.$$
Taking advantage of the relationship between cyclic codes and their dual codes about the matrix forms, we can use the same method as above to obtain that 
$$\begin{cases}
|(C^{\bot})_m|=\sum_{i=1}^{3}|((C^{\bot})_m)_{v_i}|=q^{\sum_{i=1}^{3}\deg(\theta_{v_i}(x))},        \\
|(C^{\bot})_n|=\sum_{i=1}^{3}|((C^{\bot})_n)_{v_i}|=q^{\sum_{i=1}^{3}\deg(\theta_{v_i}(x))+\sum_{i=1}^{3}k_i}. \\
\end{cases}$$
\end{proof}
\begin{corollary}
Let $C$ and $C^{\bot}$ be the above station. Then
$$
\begin{cases}
\deg(\overline{\iota}_{v_i}(x))=m-\deg(\gcd(\iota_{v_i}(x),\ell_{v_i}(x))),                                            \\
\deg(\overline{\theta}_{v_i}(x))=n-\deg(\iota_{v_i}(x))-\deg(\theta_{v_i}(x))+\deg(\gcd(\iota_{v_i}(x),\ell_{v_i}(x))), \\
\end{cases} i=1,2,3.
$$
\end{corollary}
\begin{proof}
Since $(C_m)^{\bot}$ is a cyclic code generated by $\overline{\iota}(x)$, from the conclusions about cyclic codes over $R$, we have $|((C_m)^{\bot})_{v_i}|=q^{m-\deg(\iota_{v_i}(x))}, i=1,2,3$. Furthermore, by Proposition 4, we have $|((C_m)^{\bot})_{v_i}|=q^{\deg(\iota_{v_i}(x))-k_i}, i=1,2,3$. Then $\deg(\overline{\iota}_{v_i}(x))=m-\deg(\gcd(\iota_{v_i}(x),\ell_{v_i}(x))), i=1,2,3$.

While $C^{\bot}$ is also an $R$-double cyclic code of length $(m,n)$ and $(C^{\bot})_m$ is a cyclic code generated by $\overline{\theta}(x)$. Thus $|((C^{\bot})_m)_{v_i}|=q^{n-\deg(\overline{\theta}_{v_i}(x))}$. By Proposition 4, we have $|((C^{\bot})_m)_{v_i}|=q^{\deg(\theta_{v_i}(x))+k_i}, i=1,2,3$. Consequently, $$\deg(\overline{\theta}_{v_i}(x))=n-\deg(\iota_{v_i}(x))-\deg(\theta_{v_i}(x))+\deg(\gcd(\iota_{v_i}(x),\ell_{v_i}(x))), i=1,2,3.$$
\end{proof}
\begin{example}
Let $\begin{cases}
\iota(x)=\sum_{i=0}^{4}x^i,              \\
\ell(x)=(4+2v+6v^2)x^3+(3+6v)x^2+(2+4v+4v^2)x+(5+v+2v^2),\\
\theta(x)=x+6,                            \\                            
\end{cases}$ where $\mathbb{F}_q=\mathbb{F}_7$, $m=n=5$.

Since $\ell(x)=(4x^3+3x^2+2x+5)v_1+(5x^3+2x^2+3x+1)v_2+(x^3+4x^2+2x+6)v_3$. This means that $\begin{cases}
P_{v_1}(C)=\left\langle (x^4+x^3+x^2+x+1|0), (4x^3+3x^2+2x+5|x+6)\right\rangle , \\
P_{v_2}(C)=\left\langle (x^4+x^3+x^2+x+1|0), (5x^3+2x^2+3x+1|x+6)\right\rangle , \\
P_{v_3}(C)=\left\langle (x^4+x^3+x^2+x+1|0), (1x^3+4x^2+2x+6|x+6)\right\rangle . \\
\end{cases}$ 
Therefore $P_{v_i}(C), i=1,2,3$ have the generating matrix forms of 
$$
\setlength{\arraycolsep}{2.0pt}
G_1=\left( {\begin{array}{c|c}
\begin{matrix}
1&1&1&1&1 \\
5&2&3&4&0 \\
0&5&2&3&4 \\
4&0&5&2&3 \\
3&4&0&5&2 \\
\end{matrix} &
\begin{matrix}
0&0&0&0&0 \\
6&1&0&0&0 \\
0&6&1&0&0 \\
0&0&6&1&0 \\
0&0&0&6&1 \\
\end{matrix}
\end{array}} \right),
G_2=\left( {\begin{array}{c|c}
\begin{matrix}
1&1&1&1&1 \\
1&3&2&5&0 \\
0&1&3&2&5 \\
5&0&1&3&2 \\
2&5&0&1&3 \\
\end{matrix} &
\begin{matrix}
0&0&0&0&0 \\
6&1&0&0&0 \\
0&6&1&0&0 \\
0&0&6&1&0 \\
0&0&0&6&1 \\
\end{matrix}
\end{array}} \right),
G_3=\left( {\begin{array}{c|c}
\begin{matrix}
1&1&1&1&1 \\
6&2&4&1&0 \\
0&6&2&4&1 \\
1&0&6&2&4 \\
4&1&0&6&2 \\
\end{matrix} &
\begin{matrix}
0&0&0&0&0 \\
6&1&0&0&0 \\
0&6&1&0&0 \\
0&0&6&1&0 \\
0&0&0&6&1 \\
\end{matrix}
\end{array}} \right).
$$ 
It is easy to see that both of $P_{v_i}(C),i=1,2,3$ are the optimal linear code with parameter $[10,5,5]$ over $\mathbb{F}_7$.
\end{example}
\section{Dual codes over $R$}
\begin{definition}
For $r(x) \in R[x]$, let $r(x)=\sum_{i=1}^{3}r_{v_i}(x)v_i$ with $r_{v_i}(x) \in \mathbb{F}_{q}[x], i =1,2,3$. Define the monic reciprocal polynomial of $r(x)$ is $$r^{*}(x)=\sum_{i=1}^{3}r^{*}_{v_i}(x)v_i=\sum_{i=1}^{3}(lc(r_{v_i}(x)))^{-1}x^{\deg(r_{v_i}(x))}r_{v_i}(x^{-1})v_i,$$ where $lc(r_{v_i}(x))$ expressed by the lowest term coefficient of $r_{v_i}(x)$, $i=1,2,3$.
\end{definition}
\begin{remark}
The definition of the monic reciprocal polynomial in $R[x]$ can be regarded as a generalization to the case of conventional finite fields. For this definition, we also have $(r^{*}(x))_{v_i}=(r_{v_i}(x))^{*}, i=1,2,3$. For this reason, we can write $r^{*}_{v_i}(x)$, $i=1,2,3$ without confusion.
\end{remark}
Same as Lemma 1, we also have
\begin{lemma}
Let $a(x), b(x) \in R[x]$ with $a(x)|b(x)$. Then $(\frac{b(x)}{a(x)})^{*}=\frac{b^{*}(x)}{a^{*}(x)}$.
\end{lemma}
\begin{proof}
It is equivalent to prove $(a(x)b(x))^*=a^*(x)b^*(x)$, for $a(x), b(x) \in R[x]$. From polynomial theory over traditional finite fields, we have $(f(x)g(x))^{*}=f^{*}(x)g^{*}(x)$, $f(x),g(x) \in \mathbb{F}_q[x]$. While for $r(x) \in R[x]$, we can decompose it into a combination of $\left\lbrace v_1, v_2, v_3 \right\rbrace $, which translates into the case of the polynomial over finite fields. Then we get the results in this way. 
\end{proof}
\begin{remark}
Same as the case of finite fields, we still have $r^{**}(x)=r(x), \forall r(x) \in R[x]$. In the next decomposition of polynomials, we will use this Lemma, Lemma 1 and Lemma 2 repeatedly without explanation.
\end{remark}

\begin{proposition}
Utilizing the results about Lemma 6, let 
$$C=\left\langle (\iota(x)|0),(0|\theta(x)) \right\rangle=\left\langle \left( \sum_{i=1}^3\iota_{v_i}(x)v_i|0\right) ,\left( \sum_{i=1}^3\ell_{v_i}(x)v_i| \sum_{i=1}^3 \theta_{v_i}(x)v_i\right) \right\rangle$$ 
be a separable $R$-double cyclic code of length $(m,n)$. Then $C^{\bot}$ is also a separable $R$-double cyclic code and
$$C^{\bot}=\left\langle  \left( \frac{x^m-1}{\iota^{*}(x)}|0\right),\left( 0|\frac{x^n-1}{\theta^{*}(x)}\right) \right\rangle =\left\langle \left( \sum_{i=1}^{3}\frac{x^m-1}{\iota^{*}_{v_i}(x)}v_i|0\right),\left( 0|\sum_{i=1}^{3}\frac{x^n-1}{\theta^{*}_{v_i}(x)}v_i\right) \right\rangle$$
\end{proposition}
\begin{proof}
Since $C$ is separable, we have $C=C_m \times C_n$. Thus it is easy to get $C^{\bot}=C_m^{\bot} \times C_n^{\bot}$. For the reference about cyclic codes over $\mathbb{F}_{q}+v\mathbb{F}_{q}+v^2\mathbb{F}_q$ in \cite{Melakhessou}, we can achieve that $$C^{\bot}= \left\langle  \left( \frac{x^m-1}{\iota^{*}(x)}|0\right),\left( 0|\frac{x^n-1}{\theta^{*}(x)}\right) \right\rangle =\left\langle \left( \sum_{i=1}^{3}\frac{x^m-1}{\iota^{*}_{v_i}(x)}v_i|0\right),\left( 0|\sum_{i=1}^{3}\frac{x^n-1}{\theta^{*}_{v_i}(x)}v_i\right) \right\rangle.$$   
\end{proof}

Let $\omega_m(x)$ represent the polynomial $\sum_{i=0}^{m-1} x^i$. Using this symbol, We can easily justify that
\begin{lemma}
Let $m, n \in \mathbb{N}$, then $x^{mn}-1=(x^m-1) \omega_n(x^m)$
\end{lemma} 
\begin{definition}
Let $c(x)=\left( \sum_{i=1}^{3}c^1_{v_i}(x)v_i|\sum_{i=1}^{3}c^2_{v_i}(x)v_i\right), d(x)=\left( \sum_{i=1}^{3}d^1_{v_i}(x)v_i|\sum_{i=1}^{3}d^2_{v_i}(x)v_i\right) $ be two elements in $R_{m,n}[x]$. We define the map $\circ: R_{m,n}[x] \times R_{m,n}[x] \rightarrow R_l[x]$ with $$\circ (c(x),d(x)) = \sum_{i=1}^{3}\left( c^1_{v_i}(x) \omega_{\frac{l}{m}}(x^m) x^{l-1-\deg(d^1_{v_i}(x))} d^{1\ast}_{v_i}(x)+c^2_{v_i}(x) \theta_{\frac{l}{m}}(x^m) x^{l-1-\deg(d^2_{v_i}(x))} d^{2\ast}_{v_i}(x)\right) v_i,$$ 
where the right side of equality is the combination of polynomials module $x^l-1$. 
\end{definition}

For the sake of simplicity and convenientce, we denote $\circ (c(x),d(x))$ by $c(x) \circ d(x)$. 
\begin{lemma}
Let $c,d$ be two vectors in $R^m \times R^n$, with corresponding polynomials 
$$c(x)=\left( \sum_{i=1}^{3}c^1_{v_i}(x)v_i|\sum_{i=1}^{3}c^2_{v_i}(x)v_i\right) , d(x)=\left( \sum_{i=1}^{3}d^1_{v_i}(x)v_i|\sum_{i=1}^{3}d^2_{v_i}(x)v_i\right)$$
respectively. Then $c$ is orthogonal to $d$ and all its shift if and only if $c(x) \circ d(x)  \equiv 0$.
\end{lemma}
\begin{proof}
Let $d_{(s)}=(d^1_{0+s},\cdots,d^1_{m-1+s} \mid d^2_{0+s},\cdots,d^2_{n-1+s})$ be the $s$-th cyclic shift of vector $d$, $0 \le s \le l-1$. We know that $\left\langle c,d_{(s)}\right\rangle =0$ if and only if $\sum_{k_1=0}^{m-1} c^1_{k_1} d^1_{k_1+s} +\sum_{k_2=0}^{n-1} c^2_{k_2} d^2_{k_2+s}=0$. Fix $\mathsf{\Delta}_s = \sum_{k_1=0}^{m-1} c^1_{k_1} d^1_{k_1+s} +\sum_{k_2=0}^{n-1} c^2_{k_2} d^2_{k_2+s}$, we can get that
\begin{equation} \begin{split}
c(x) \circ d(x) & = \sum_{i=0}^{m-1} \left( \omega_{\frac{l}{m}}(x^m) \sum_{k_1=0}^{m-1} c^1_{k_1} d^1_{k_1+i} x^{l-1-i}\right) +\sum_{j=0}^{n-1}\left( \omega_{\frac{l}{n}}(x^n) \sum_{k_2=0}^{n-1} c^2_{k_2} d^2_{k_2+j} x^{l-1-j}\right)  \notag \\
& = \omega_{\frac{l}{m}}(x^m) \left[ \sum_{i=0}^{m-1} \sum_{k_1=0}^{m-1}c^1_{k_1}d^1_{k_1+i}x^{l-1-i}\right] +\omega_{\frac{l}{n}}(x^n)\left[ \sum_{j=0}^{m-1} \sum_{k_2=0}^{n-1}c^2_{k_2} d^2_{k_2+j}x^{l-1-j}\right] \notag \\
&= \sum_{s=0}^{l-1} \mathsf{\Delta}_s x^{l-1-s} \notag \\
\end{split} \end{equation}
in $R[x]/(x^l-1)$. Hence $c(x) \circ d(x)=0$ if and only if $\mathsf{\Delta}_s=0$ for all $0 \le s \le l-1$.
\end{proof}
\begin{lemma}
Let $c(x)=(c^1(x)|c^2(x))$, $d(x)=(d^1(x)|d^2(x))$ are two elements in $R_{m,n}[x]$, such that $c(x) \circ d(x) =0$ mod $(x^l-1)$. If $c^1(x) \equiv 0$ or $d^1(x) \equiv 0$, then $c^2(x)d^{2\ast}(x)=0$ mod $(x^n-1)$. Respectively, if $c^2(x) \equiv 0$ or $d^2(x) \equiv 0$, then $c^1(x)d^{1\ast}(x)=0$ mod $(x^m-1)$.
\end{lemma}
\begin{proof}
Let $c^2(x)$ or $d^2(x)$ equal to 0  module $x^n-1$. It means that $c^2_{v_i}(x) \equiv 0$, $i=1,2,3$ or $d^2_{v_i}(x) \equiv 0$, $i=1,2,3$. By the concrete definition forms of $\circ$ in $R[x]$, we have  $$c(x) \circ d(x) = \sum_{i=1}^{3}\left( c^1_{v_i}(x) \omega_{\frac{l}{m}}(x^m) x^{l-1-\deg(d^1_{v_i}(x))} d^{1\ast}_{v_i}(x)\right) v_i=0 \mod (x^l-1).$$
Hence there exists a polynomial $\pi(x) \in R[x]$ with $\pi(x)=\sum_{i=1}^{3}\pi_{v_i}(x)v_i$, such that $$\sum_{i=1}^{3}\left( c^1_{v_i}(x) \omega_{\frac{l}{m}}(x^m) x^{l-1-\deg(d^1_{v_i}(x))} d^{1\ast}_{v_i}(x)\right) v_i=\pi(x)(x^l-1)=\sum_{i=1}^{3}\pi_{v_i}(x)(x^l-1)v_i.$$
Then $c^1_{v_i}(x) \omega_{\frac{l}{m}}(x^m) x^{l-1-\deg(d^1_{v_i}(x))} d^{1\ast}_{v_i}(x)=\pi_{v_i}(x)(x^l-1), i=1,2,3.$ Let $\Pi(x)=\sum_{i=1}^{3}\Pi_{v_i}(x)v_i$ with $\Pi_{v_i}(x)=x^{\deg(d^1_{v_i}(x))+1}\pi_{v_i}(x), i=1,2,3$, we have $c^1_{v_i}(x) \omega_{\frac{l}{m}}(x^m) x^l d^{1\ast}_{v_i}(x)=\Pi_{v_i}(x)(x^l-1), i=1,2,3.$ By Lemma 8, we get $x^l-1=\omega_{\frac{l}{m}}(x^m)(x^m-1)$. So $c^1_{v_i}(x) d^{1\ast}_{v_i}(x) x^l =\Pi_{v_i}(x)(x^m-1) \text{ } i=1,2,3.$ It means that $x^m-1|c^1_{v_i}(x) d^{1\ast}_{v_i}(x) x^l, i=1,2,3.$ Obviously, $x^m-1$ and $x^l$ are co-prime to each other. This leads to $x^m-1|c^1_{v_i}(x) d^{1\ast}_{v_i}(x), i=1,2,3$.
Hence $$\sum_{i=1}^{3}c^1_{v_i}(x)d^{1\ast}_{v_i}(x)v_i=c^1(x)d^{1\ast}(x)=0 \mod (x^m-1).$$ The same assertion can be proved for the other cases.
\end{proof}

\begin{proposition}
Let $\thickmuskip=0mu \medmuskip=0mu C=\left\langle (\iota(x)|0),(\ell(x)|\theta(x)) \right\rangle =\left\langle (\sum_{i=1}^3\iota_{v_i}(x)v_i|0),(\sum_{i=1}^3\ell_{v_i}(x)v_i| \sum_{i=1}^3 \theta_{v_i}(x)v_i)\right\rangle$ be a double cyclic code of length $(m,n)$ over $R$ with $$C^{\bot}=\left\langle (\overline{\iota}(x)|0),(\overline{\ell}(x)|\overline{\theta}(x)) \right\rangle =\left\langle \left( \sum_{i=1}^3\overline{\iota}_{v_i}(x)v_i|0\right),\left( \sum_{i=1}^3\overline{\ell}_{v_i}(x)v_i| \sum_{i=1}^3 \overline{\theta}_{v_i}(x)v_i\right) \right\rangle.$$	
Then $\overline{\iota}(x)=\frac{x^m-1}{\gcd^{*}(\iota(x),\ell(x))}=\sum_{i=1}^{3}\frac{x^m-1}{\gcd^{*}(\iota_{v_i}(x),\ell_{v_i}(x))}v_i$.
\end{proposition}
\begin{proof}
Obviously, $(\overline{\iota}(x)|0)$ belongs to $C^{\bot}$. As a consequence of Lemma 9, we have
$$\begin{cases}
(\overline{\iota}(x)|0) \circ (\iota(x)| 0)       &=0 \mod (x^l-1), \\
(\overline{\iota}(x)|0) \circ (\ell(x)|\theta(x)) &=0 \mod (x^l-1). \\
\end{cases}$$
Therefore, by Lemma 10, we also have
$$\begin{cases} 
\overline{\iota}^{*}(x) \iota(x) &=0 \mod (x^m-1) \iff (x^m-1)|\overline{\iota}^{*}(x) \iota(x),\\ 
\overline{\iota}^{*}(x) \ell(x) &=0 \mod (x^m-1) \iff (x^m-1)|\overline{\iota}^{*}(x) \ell(x).  \\ 
\end{cases}$$ 
Then we get $x^m-1|\gcd(\overline{\iota}^{*}(x) \iota(x),\overline{\iota}^{*}(x) \ell(x))|\overline{\iota}^{*}(x) \gcd(\iota(x),\ell(x))$. While $x^m-1|\overline{\iota}^{*}(x) \gcd(\iota(x),\ell(x))$ if and only if 
$$x^m-1|\overline{\iota}^{*}_{v_i}(x) {\gcd}_{v_i}(\iota(x),\ell(x))=\overline{\iota}^{*}_{v_i}(x)\gcd(\iota_{v_i}(x),\ell_{v_i}(x)), i=1,2,3.$$
Since all of $\overline{\iota}^{*}_{v_i}(x), \gcd(\iota_{v_i}(x),\ell_{v_i}(x))$, $i=1,2,3$ are factors of $x^m-1$, By Corollary 2, we have $$\deg(\overline{\iota}^{*}_{v_i}(x))=\deg(\overline{\iota}_{v_i}(x))=m-\deg(\gcd(\iota_{v_i}(x),\ell_{v_i}(x))), i=1,2,3.$$
Therefore $x^m-1=\overline{\iota}^{*}_{v_i}(x) \gcd_{v_i}(\iota(x),\ell(x))=\overline{\iota}^{*}_{v_i}(x)\gcd(\iota_{v_i}(x),\ell_{v_i}(x))$, $i=1,2,3$.
Then we get $$\overline{\iota}^{*}(x)\gcd(\iota(x),\ell(x))=\sum_{i=1}^{3}\overline{\iota}^{*}_{v_i}(x)\gcd(\iota_{v_i}(x),\ell_{v_i}(x))v_i=\sum_{i=1}^{3}(x^m-1)v_i=x^m-1.$$ Consequently, we have $\overline{\iota}(x) = \frac{x^m-1}{\gcd^{*}(\iota(x),\ell(x))}=\sum_{i=1}^{3}\frac{x^m-1}{\gcd^{*}(\iota_{v_i}(x),\ell_{v_i}(x))}v_i$.
\end{proof}
\begin{proposition}
Let $\thickmuskip=0mu \medmuskip=0mu C=\left\langle (\iota(x)|0),(\ell(x)|\theta(x)) \right\rangle =\left\langle (\sum_{i=1}^3\iota_{v_i}(x)v_i|0),(\sum_{i=1}^3\ell_{v_i}(x)v_i| \sum_{i=1}^3 \theta_{v_i}(x)v_i)\right\rangle$ be an $R$-double cyclic code of length $(m,n)$ with $$C^{\bot}=\left\langle (\overline{\iota}(x)|0),(\overline{\ell}(x)|\overline{\theta}(x)) \right\rangle =\left\langle \left( \sum_{i=1}^3\overline{\iota}_{v_i}(x)v_i|0\right),\left( \sum_{i=1}^3\overline{\ell}_{v_i}(x)v_i| \sum_{i=1}^3 \overline{\theta}_{v_i}(x)v_i\right) \right\rangle.$$
Then $\overline{\theta}(x) = \frac{(x^n-1)\gcd^{*}(\iota(x),\ell(x))}{\iota^{*}(x)\theta^{*}(x)}=\sum_{i=1}^{3}\frac{(x^n-1)\gcd^{*}(\iota_{v_i}(x),\ell_{v_i}(x))}{\iota_{v_i}^{*}(x)\theta_{v_i}^{*}(x)}v_i$.
\end{proposition} 
\begin{proof}
Considering the codeword  
$$\left( 0|\frac{\iota(x)}{\gcd(\iota(x),\ell(x))}\theta(x) \right) =\frac{\iota(x)}{\gcd(\iota(x),\ell(x))}(\ell(x)|\theta(x))-\frac{\ell(x)}{\gcd(\iota(x),\ell(x))}(\iota(x)|0),$$
we have $\left( 0|\frac{\iota(x)}{\gcd(\iota(x),\ell(x))}\theta(x) \right) \in C$. By Lemma 9, we get 
$$(\overline{\ell}(x)|\overline{\theta}(x)) \circ \left( 0|\frac{\iota(x)}{\gcd(\iota(x),\ell(x))}\theta(x) \right) = 0 \mod (x^l-1).$$ 
Then, from Lemma 10, we obtain that $$\overline{\theta}(x) \frac{\iota^{*}(x)\theta^{*}(x)}{\gcd^{*}(\iota(x),\ell(x))} = 0 \mod (x^n-1) \iff x^n-1|\overline{\theta}(x) \frac{\iota^{*}(x)\theta^{*}(x)}{\gcd^{*}(\iota(x),\ell(x))}.$$
While $x^n-1|\overline{\theta}(x) \frac{\iota^{*}(x)\theta^{*}(x)}{\gcd*(\iota(x),\ell(x))}$ if and only if 
$$x^n-1|\overline{\theta}_{v_i}(x) \frac{\iota^{*}_{v_i}(x)\theta^{*}_{v_i}(x)}{\gcd^{*}_{v_i}(\iota(x),\ell(x))}=\overline{\theta}_{v_i}(x) \frac{\iota^{*}_{v_i}(x)\theta^{*}_{v_i}(x)}{\gcd^{*}(\iota_{v_i}(x),\ell_{v_i}(x))}, i=1,2,3.$$
From Theorem 3, we acquire $\overline{\theta}_{v_i}(x)|(x^n-1), i=1,2,3$. Simultaneously, by Lemma 5, we have that $\frac{\iota^{*}_{v_i}(x)\theta^{*}_{v_i}(x)}{\gcd(\iota_{v_i}(x),\ell_{v_i}(x))}|(x^n-1), i =1,2,3$.
From Corollary 2, we gain $$\deg(\overline{\theta}_{v_i}(x))=n-\deg(\iota_{v_i}(x))-\deg(\theta_{v_i}(x))+\deg(\gcd(\xi_{v_i}(x),\ell_{v_i}(x)), i=1,2,3.$$
Therefore $\deg(\frac{\overline{\theta}_{v_i}(x)\iota^{*}_{v_i}(x)\theta^{*}_{v_i}(x)}{\gcd^{*}(\iota_{v_i}(x),\ell_{v_i}(x))})=n, i=1,2,3$. These indicate that $x^n-1=\frac{\overline{\theta}_{v_i}(x)\iota^{*}_{v_i}(x)\theta^{*}_{v_i}(x)}{\gcd^{*}(\iota_{v_i}(x),\ell_{v_i}(x))}, i=1,2,3$.
Hence $$\overline{\theta}(x) \frac{\iota^{*}(x)\theta^{*}(x)}{\gcd^{*}(\iota(x),\ell(x))}= \sum_{i=1}^{3}\overline{\theta}_{v_i}(x) \frac{\iota^{*}_{v_i}(x)\theta^{*}_{v_i}(x)}{\gcd^{*}(\iota_{v_i}(x),\ell_{v_i}(x))}v_i=\sum_{i=1}^{3}(x^n-1)v_i=(x^n-1)\sum_{i=1}^{3}v_i=x^n-1.$$
Thereby $\overline{\theta}(x)=\frac{(x^n-1)\gcd^{*}(\iota(x),\ell(x))}{\iota^{*}(x)\theta^{*}(x)}=\sum_{i=1}^{3}\frac{(x^n-1)\gcd^{*}(\iota_{v_i}(x),\ell_{v_i}(x))}{\iota^{*}_{v_i}(x)\theta^{*}_{v_i}(x)}v_i.$
\end{proof}
\begin{remark}
We use the fact that $\deg(f^{*}(x))=\deg(f(x))$, for $f(x) \in \mathbb{F}_{q}[x]$ in the above proof.
\end{remark}
\begin{proposition}
Let $\thickmuskip=0mu \medmuskip=0mu C=\left\langle (\iota(x)|0),(\ell(x)|\theta(x)) \right\rangle =\left\langle (\sum_{i=1}^3\iota_{v_i}(x)v_i|0),(\sum_{i=1}^3\ell_{v_i}(x)v_i| \sum_{i=1}^3 \theta_{v_i}(x)v_i)\right\rangle$ be a double cyclic code of length $(m,n)$ over $R$ with $$C^{\bot}=\left\langle (\overline{\iota}(x)|0),(\overline{\ell}(x)|\overline{\theta}(x)) \right\rangle =\left\langle \left( \sum_{i=1}^3\overline{\iota}_{v_i}(x)v_i|0\right),\left( \sum_{i=1}^3\overline{\ell}_{v_i}(x)v_i| \sum_{i=1}^3 \overline{\theta}_{v_i}(x)v_i\right) \right\rangle.$$
Then $\overline{\ell}(x) = \rho(x) \left( \sum_{i=1}^{3}\frac{x^m-1}{\iota^{*}_{v_i}(x)}v_i\right)$, where 
$$\rho(x) = \left( \sum_{i=1}^{3}-x^{l-\deg(\theta_{v_i}(x))+\deg(\iota_{v_i}(x))}v_i\right) \left( \frac{\iota^{*}(x)}{\gcd^{*}(\iota(x),\ell(x))}\right)^{-1} \mod\frac{\iota^{*}(x)}{\gcd^{*}(\iota(x),\ell(x))}.$$
\end{proposition}
\begin{proof}
Since $(\overline{\ell}(x)|\overline{\theta}(x)) \in C^{\bot}$, $(\iota(x)|0) \in C$, we have
$$(\overline{\ell}(x)|\overline(\theta)(x)) \circ (\iota(x)|0) \equiv 0 \mod (x^l-1)$$ from Lemma 9.
Then we get $\overline{\ell}(x)\iota^{*}(x) = 0 \mod (x^m-1)$ in view of Lemma 10. Thus there exists a polynomial $\rho(x) \in R[x]$ such that $\overline{\ell}(x)=\frac{x^m-1}{\iota^{*}(x)} \rho(x) = (\sum_{i=1}^{3}\frac{x^m-1}{\iota^{*}_{v_i}(x)}v_i) \rho(x)$. We explain the details of $\rho(x)$ in the following proof.

From Lemma 9, we have $(\overline{\ell}(x) \mid \overline{\theta}(x)) \circ (\ell(x)|\theta(x))=0 \mod (x^l-1)$. Writing the concrete expression of $(\overline{\ell}(x) \mid \overline{\theta}(x)) \circ (\ell(x)|\theta(x))$, we obtain that
\begin{equation} \begin{split}
&(\overline{\ell}(x)|\overline{\theta}(x)) \circ (\ell(x)|\theta(x))=\left( \frac{x^m-1}{\iota^{*}(x)} \rho(x)|\frac{(x^n-1)\gcd^{*}(\iota(x),\ell(x))}{\iota^{*}(x)\theta^{*}(x)}\right)  \circ (\ell(x)|\theta(x))= \sum_{i=1}^{3}\left( (\frac{x^m-1}{\iota^{*}_{v_i}(x)} \right. \\ &\left.
\rho(x) \omega_{\frac{l}{m}}(x^m) x^{l-1-\deg(\ell_{v_i}(x))} \ell^{*}_{v_i}(x) + \frac{(x^n-1) \gcd^{*}_{v_i}(\iota(x),\ell(x))}{\iota^{*}_{v_i}(x) \theta^{*}_{v_i}(x)} \omega_{\frac{l}{n}}(x^n) x^{l-1-\deg(\theta_{v_i}(x))} \theta^{*}_{v_i}(x)\right)v_i. \notag \\
\end{split} \end{equation}
And $(x^m-1) \omega_{\frac{l}{m}}(x^m)=x^l-1$, $(x^n-1) \omega_{\frac{l}{n}}(x^n)=x^l-1$, we have
$$\thickmuskip=0mu \medmuskip=0mu \sum_{i=1}^{3}\frac{(x^l-1)\gcd^{*}_{v_i}(\iota(x),\ell(x))}{\iota^{*}_{v_i}(x)}\left( \rho_{v_i}(x)x^{l-\deg(\ell_{v_i}(x))}\frac{\ell^{*}_{v_i}(x)}{\gcd^{*}_{v_i}(\iota(x),\ell(x))}+x^{l-\deg(\theta_{v_i}(x))-1}\right)v_i= 0 \text{mod}(x^l-1).$$
This reveals that
$$\thickmuskip=0mu \medmuskip=0mu \sum_{i=1}^{3}\frac{x^l-1}{\iota^{*}_{v_i}(x)/\gcd^{*}_v(\iota(x),\ell(x))} \left( \rho_{v_i}(x) x^{l-\deg(\ell_{v_i}(x))} \frac{\ell^{*}_{v_i}(x)}{\gcd^{*}_{v_i}(\iota(x),\ell(x))} + x^{l-\deg(\theta_{v_i}(x))-1}\right)v_i = 0 \text{mod}(x^l-1).$$
Set $\tilde{\iota}(x)=\frac{\iota(x)}{\gcd(\iota(x),\ell(x))},\tilde{\ell}(x)=\frac{\ell(x)}{\gcd(\iota(x),\ell(x))}$,
hence $$\sum_{i=1}^{3}\frac{x^l-1}{\tilde{\iota}^{*}_{v_i}(x)}\left( \rho_{v_i}(x)x^{l-\deg(\ell_{v_i}(x))}\tilde{\ell}^{*}_{v_i}(x)+x^{l-\deg(\theta_{v_i}(x))-1}\right)v_i = 0 \mod(x^l-1).$$
Therefore $$\sum_{i=1}^{3}\left( \rho_{v_i}(x)x^{l-\deg(\ell_{v_i}(x))}\tilde{\ell}^{*}_{v_i}(x)+x^{l-\deg(\theta_{v_i}(x))-1}\right)v_i = 0 \mod (x^l-1)$$ or
$$\sum_{i=1}^{3}\left( \rho_{v_i}(x)x^{l-\deg(\ell_{v_i}(x))}\tilde{\ell}^{*}_{v_i}(x)+x^{l-\deg(\theta_{v_i}(x))-1}\right)v_i= 0 \mod (\tilde{\iota}^{*}(x)).$$
Since the former can be deduced the latter by reason of $\tilde{\iota}^{*}(x)|(x^l-1)$, we can assume that 
$$\sum_{i=1}^{3}\left( \rho_{v_i}(x)x^{l-\deg(\ell_{v_i}(x))}\tilde{\ell}^{*}_{v_i}(x)+x^{l-\deg(\theta_{v_i}(x))-1}\right)v_i= 0 \mod (\tilde{\iota}^{*}(x)).$$
From the fixing of $\tilde{\iota}(x)$ and $\tilde{\ell}(x)$, it is obviously that $\gcd(\tilde{\iota}(x),\tilde{\ell}(x))=1$. Furthermore $x^l = 1 \mod \iota^{*}(x)$. Then $\tilde{\ell}^{*}(x)$ is an invertible element modulo $\tilde{\iota}^{*}(x)$. Consequently, we have
$$\rho(x) = \left( \sum_{i=1}^{3}-x^{l-\deg(\theta_{v_i}(x))+\deg(\iota_{v_i}(x))}v_i\right) \left( \frac{\ell^{*}(x)}{\gcd^{*}(\iota(x),\ell(x))}\right) ^{-1} \mod\frac{\iota^{*}(x)}{\gcd^{*}(\iota(x),\ell(x))}.$$
\end{proof}
Summarizing several propositions and lemmas, we get the second primary theorem of this article.
\begin{theorem}
Let $C$ be an $R$-double cyclic code of length $(m,n)$, and set $C,C^{\bot}$ have the forms of $$\begin{cases}
C&=\left\langle (\iota(x)|0),(\ell(x)|\theta(x)) \right\rangle =\left\langle \left( \sum_{i=1}^3\iota_{v_i}(x)v_i|0\right),\left( \sum_{i=1}^3\ell_{v_i}(x)v_i| \sum_{i=1}^3 \theta_{v_i}(x)v_i\right) \right\rangle  \\
C^{\bot}&=\left\langle (\overline{\iota}(x)|0),(\overline{\ell}(x)|\overline{\theta}(x)) \right\rangle =\left\langle \left( \sum_{i=1}^3\overline{\iota}_{v_i}(x)v_i|0\right),\left( \sum_{i=1}^3\overline{\ell}_{v_i}(x)v_i| \sum_{i=1}^3 \overline{\theta}_{v_i}(x)v_i\right) \right\rangle.  \\
\end{cases}$$
Then
	
(1) $\overline{\iota}(x)=\sum_{i=1}^{3}\frac{x^m-1}{\gcd^{*}(\iota_{v_i}(x),\ell_{v_i}(x))}v_i$
	
(2) $\overline{\theta}(x)=\sum_{i=1}^{3}\frac{(x^n-1)\gcd^{*}(\iota_{v_i}(x),\ell_{v_i}(x))}{\iota^{*}_{v_i}(x)\theta^{*}_{v_i}(x)}v_i$
	
(3) $\overline{\ell}(x) =\rho(x) \sum_{i=1}^{3} \frac{x^m-1}{\iota^{*}_{v_i}(x)}v_i$, where
$$\begin{cases}
\rho(x)=0 \text{ } \text{if C is separable, or otherwise} \\
\rho(x)=\left( \sum_{i=1}^{3}-x^{l-\deg(\theta_{v_i}(x))+\deg(\iota_{v_i}(x))}v_i\right) \left( \frac{\iota^{*}(x)}{\gcd^{*}(\iota(x),\ell(x))}\right)^{-1} \mod\frac{\iota^{*}(x)}{\gcd^{*}(\iota(x),\ell(x))}\\
\end{cases}$$
\end{theorem}

\section{Conclusion}
In this paper, we analyze the algebraic structure of double cyclic codes over $\mathbb{F}_q+v\mathbb{F}_q+v^2\mathbb{F}_q$ with $v^3=v$. We provide the generator polynomials of these series of codes and give their generating matrices forms. Furthermore, we also discuss the quantitive relationship between the double cyclic codes and their duals. Finally, we determine the relationship between the generator of double cyclic codes and their duals.

Specially, if we let $v=0 \text{ } \text{or} \text{ } 1 \text{ } \text{or} \text{ } -1$, then $v_1=1,v_2=v_3=0$ or $v_2=1,v_1=v_3=0$ or $v_3=1, v_1=v_2=0$. The article show up the results about double cyclic codes over $\mathbb{F}_q$. This is the content of the paper \cite{Diao}.

\end{document}